\documentclass[journal]{IEEEtran}
\usepackage{amsthm,amsmath,amssymb}
\usepackage{mathrsfs}
\usepackage{amsfonts}
\usepackage{graphicx, amssymb}
\usepackage{amsthm}
\usepackage{amsmath}
\usepackage{epsfig}
\usepackage{graphicx}
\usepackage{graphics}
\usepackage{multirow}
\usepackage{threeparttable}
\usepackage{amsthm}
\usepackage{float}
\usepackage{color}
\usepackage{algorithm}
\usepackage{algpseudocode}
\usepackage{bm}
\usepackage{algorithmicx}  
\usepackage{algpseudocode}
\usepackage{setspace}
\usepackage{authblk} 
\usepackage{cite}
\usepackage{pifont}
\usepackage{xcolor}
\theoremstyle{definition}
\newtheorem{theorem}{Theorem} 

\newtheorem{lemma}{Lemma}

\newcommand{\RNum}[1]{\uppercase\expandafter{\romannumeral #1\relax}}

\ifCLASSINFOpdf
\else
\fi

\begin{document}

\title{Rate-Region Characterization and Channel Estimation for  Cell-Free Symbiotic Radio Communications  }

%

\author{
   {Zhuoyin Dai, Ruoguang Li, Jingran Xu, 
   Yong Zeng,~\IEEEmembership{Member, IEEE}, 
   and Shi Jin,~\IEEEmembership{Senior Member, IEEE} }  \\

\thanks{
This work was supported by the National Key R\&D Program of China with 
grant number 2019YFB1803400.
Part of this work has been presented at 
the  2021 IEEE/CIC ICCC Workshops, Xiamen, China, 28-30 Jul.  2021 \cite{Dai2021a}.

The authors are with the National Mobile Communications Research Laboratory, 
Southeast University, Nanjing 210096, China. Y. Zeng is also with the 
Purple Mountain Laboratories, Nanjing 211111, 
China (e-mail: \{zhuoyin\_dai, ruoguangli, jingran\_xu, yong\_zeng, jinshi\}@seu.edu.cn). 
(\emph{Corresponding author: Yong Zeng.})
}
}
\maketitle


\begin{abstract}
 
Cell-free massive MIMO and symbiotic radio communication
have been recently proposed as the
promising beyond fifth-generation (B5G)  networking architecture and transmission
technology, respectively. To reap the benefits of both, 
this  paper studies cell-free symbiotic
radio communication systems, where a number of cell-free access points
(APs) cooperatively send primary information to a receiver,
and simultaneously support the passive backscattering communication
of the secondary backscatter device (BD). 
We first derive the  achievable communication  rates 
of the active primary user and passive secondary  user
under the assumption of perfect  channel state information (CSI),
based on which the  transmit beamforming of the cell-free APs 
is optimized to characterize the achievable rate-region of 
cell-free symbiotic communication systems. 
Furthermore, to practically acquire the CSI of the 
active and passive channels, 
we propose  an efficient channel estimation method based on two-phase
uplink-training, and the 
achievable rate-region 
taking into account CSI estimation errors are further characterized.
Simulation results are provided to show the effectiveness 
of our proposed  beamforming and channel estimation methods. 
\end{abstract}


\IEEEpeerreviewmaketitle

\section{Introduction}
With the  ongoing  commercial deployment of the fifth-generation (5G) mobile communication 
networks, the academia
and industry communities have started the investigation of the key technologies 
for  beyond fifth-generation (B5G) or the sixth-generation (6G) networks \cite{LavtaM2019a, YongZ2021a,YouX2021a}. 
In order to meet the orders-of-magnitude performance improvement  in terms of coverage, 
connectivity density, data rate, reliability,  latency, etc., many 
promising technologies are being investigated, such as
extremely large-scale MIMO/surface \cite{HuS2018a,LuH2021a},
millimeter wave or TeraHertz communication \cite{MR2014a,ElayanH2018a }, 
non-terrestrial networks (NTN) \cite{GiordaniM2021a, YZeng2019a }, 
reconfigurable intelligent surface (RIS) \cite{TangW2019a,WuQ2021b},
and artificial intelligence (AI)-aided wireless communications \cite{LuongNC2019a}.
On the other hand,
\textit{cell-free massive MIMO} \cite{Ngo2017a} and \emph{symbiotic radio communication} \cite{YLiang2020a} 
were recently proposed as the promising B5G networking 
architecture and transmission technology, respectively,
which have received fast-growing attentions.

As a radically  new potential networking architecture  for B5G mobile 
communication networks, cell-free massive MIMO is significantly  different from the 
classical cellular  
architecture since it blurs  the conventional concepts of 
cells or cell boundaries \cite{Ngo2017a}. Instead, geographically  distributed 
access points (APs) \cite{BjornsonE2020a,ChenZ2018a}, which are connected to the central processing unit (CPU), 
cooperatively serve their surrounding users to achieve high macro diversity.  
Cell-free massive MIMO  is expected to mitigate the inter-cell
interference issues suffered by  small cell 
systems and  provide users with consistently high quality of service 
everywhere \cite{Interdonato2019a}. 
Significant  research efforts have been recently devoted to
the theoretical study and practical design of cell-free massive MIMO systems. 
For example,  the performance of two basic linear precoding schemes, i.e.,
conjugate beamforming and zero-forcing precoding,  
was compared for cell-free massive MIMO in  \cite{Nayebi2017a}. 
The receiver filter coefficients   and power allocation of 
cell-free massive MIMO were optimized  to maximize
the minimal user rate or bandwidth efficiency in \cite{BasharM2019a,BasharM2020a}. 
Furthermore, in \cite{NguyenL2017a,Ngo2018a},   
the communication resource allocation was optimized to maximize the
energy efficiency and spectral efficiency of cell-free massive MIMO systems.

On the other hand, symbiotic radio has been recently proposed 
as a promising B5G transmission 
technology \cite{YLiang2020a}, which  is able to exploit the benefits of the 
conventional cognitive radio (CR) and the emerging passive ambient backscattering 
communications (AmBC) to realize spectral- and 
energy-efficient communications \cite{RLong2020a}. Specifically, 
the passive secondary backscatter device (BD) in symbiotic radio systems 
reuses  not only the spectrum of the active
primary communication as in traditional CR systems, but also its
power via passive backscattering technology  \cite{RLong2019a}.  
Based on the relationship of symbol durations of the primary
and secondary signals,  symbiotic radio systems can be classified as
\emph{commensal symbiotic radio} (CSR) and \emph{parasite symbiotic radio} (PSR) \cite{GuoH2019b}. 
In CSR, the secondary signals have  much longer symbol durations
than the primary signals, rendering  the secondary backscattering communication to
contribute additional multipath components to enhance the primary communication. 
As a result, the primary and secondary communications form a mutualism 
relationship \cite{RLong2020a}.  On the other hand, for PSR, 
the primary and secondary signals  have  equal symbol durations,
so that the secondary signals may interfere with the primary signal. 
However, compared to the CSR case, the secondary communication rate in PSR can be significantly improved.  
Significant research efforts have been devoted to the study 
of symbiotic radio systems. For example, in order to 
maximize the secondary communication rate, 
an exact penalty beamforming method based on the local optimal solution 
was proposed in \cite{WuT2021a}. Besides, \cite{GuoH2019a} and \cite{ChuZ2020a} 
investigated how to effectively allocate communication resources such 
as transmit power and reflection coefficient, so as to 
improve the energy efficiency and achievable rates of symbiotic radio.


It is worth remarking that all the aforementioned 
existing  works studied cell-free massive MIMO or 
symbiotic radio communication systems separately, i.e., 
cell-free  systems with   conventional 
active communication or symbiotic radio transmission in conventional 
cellular network or  the simplest  point-to-point communications.
As the promising 
B5G  networking architecture and transmission technology, respectively, it is 
natural that cell-free networking and symbiotic radio communication would 
merge into each other to reap the benefits of both. 
This motivates our current 
work to investigate  cell-free symbiotic
radio communication systems, which, to the best of our knowledge, 
have not been studied in the existing literature.  
By combining cell-free architecture with symbiotic radio transmission technology, 
the passive secondary communication in symbiotic radio system 
is enhanced by the cooperation gain of distributed APs,  
thus realizing passive communication with high  macro-diversity.   
In this paper, we study a basic cell-free symbiotic radio 
system, in which a number of distributed multi-antenna  APs  
cooperatively send primary information to a receiver, and concurrently support 
the passive backscattering communication of the secondary BD.
As such,  the distributed cooperation gain by APs can be exploited to 
enhance both the primary and secondary communication rates.
Our specific contributions are summarized as follows:
\begin{itemize}
	\item First, we present  the mathematical  model of cell-free symbiotic
radio communication systems, which is a promising  
system  that exploits both advantages of 
cell-free networking architecture and symbiotic radio transmission technology.  
Under the assumption of perfect channel state information (CSI) of the direct active channels 
and cascaded passive channels, the
achievable rates of both the primary and secondary communications are derived.
	
\item Next, we relax the assumption of perfect CSI and investigate the practical 
CSI acquisition method for the considered cell-free symbiotic radio system.  
Similar to the extensively studied massive MIMO systems, efficient channel 
estimation for cell-free massive MIMO can
be achieved by exploiting the uplink-downlink channel reciprocity 
\cite{MishraD2019a, KaltenbergerF2010a,MarzettaTL2010a}, 
i.e., the downlink channels can be efficiently
estimated via uplink training.   However, different from the
existing cell-free massive MIMO systems \cite{Ngo2017a}, the channel
estimation for cell-free symbiotic radio system requires estimating
not only the active direct-link channels, but also the passive backscatter
channels. To this end, we propose a two-phase based channel estimation method
for cell-free symbiotic
radio communication systems. In the first phase, pilot
symbols are sent by the receiver while muting the BD, so as
to estimate the direct-link channels. In the
second phase, pilots are sent  by both the receiver and the BD
so that, together with the estimation of the direct-link channels,
the cascaded backscatter channels are estimated. 
Furthermore, the  channel estimation errors in both phases are derived,
which are shown to be dependent on the total pilot length and the pilot allocation 
between the two training phases.
The achievable rates 
under  imperfect CSI are derived 
by taking into account the CSI estimation errors.

	\item  Furthermore, for both the ideal scenario with perfect CSI 
and practical scenario of 
imperfect CSI with channel estimation errors, we formulate the beamforming 
optimization problem to
characterize the achievable rate-region of the active primary communication and 
the passive secondary communication. The formulated problems are non-convex  
in general, which are difficult to be directly solved. 
We show that a closed-form solution can be obtained for the 
special case when the targeting primary rate is relatively small. 
Furthermore, for the general cases,  we show  that the rate threshold constraint 
can be converted into the convex second-order cone (SOC) constraint,
and that the nonconcave  objective function can be globally lower-bounded  
by its first-order Taylor expansion.
Therefore, efficient algorithms are proposed 
based on   successive convex approximation (SCA) technique 
\cite{zeng2017energy,Zappone2017a,Marks1978a}. 
Numerical  results are provided to demonstrate that
the proposed channel estimation and optimization approaches are effective 
in cell-free symbiotic radio communication systems.

\end{itemize}

The rest of this paper is organized as follows. Section II
presents the mathematical  model  of cell-free symbiotic radio communication  systems.
Under the assumption of perfect CSI, Section III 
characterizes the 
achievable rate-region of passive secondary communication and 
active primary communication
by optimizing the transmit beamforming of the APs.
In Section IV, 
a two-phase uplink-training 
based channel estimation method 
is proposed,  and the achievable primary 
and secondary communication rates taking into account the channel estimation errors are 
derived. Furthermore, the beamforming optimization problem with 
imperfect CSI is also studied in Section IV. 
Section V presents numerical results to validate our proposed designs.  
Finally, we conclude the paper in Section VI.

\emph{Notations:} In this paper, scalars are denoted by italic letters.
Vectors and matrices are denoted by boldface lower-and upper-case letters
respectively. $\mathbb{C}^{N\times 1}$ denotes the
space of $N$-dimensional complex-valued vectors.
$\mathrm{Re}\{\cdot\}$ and $\mathrm{Im}\{\cdot\}$ denote the   
real and imaginary parts, respectively.
$\mathbb{E}_{X}[\cdot]$ denotes the expectation with respect to the
random variable $X$.  
$\mathrm{Ei}(x)\triangleq \int_{-\infty}^{x}\frac{1}{u}e^{u}du$ denotes 
the exponential integral from $-\infty $ to $x$.
 $\mathbf I_N$ denotes an $N \times N$ identity matrix.
For a vector $\mathbf{a}$, its
transpose, Hermitian transpose, and Euclidean norm are respectively
denoted as
$\mathbf{a}^{T}$, $\mathbf{a}^{H}$ and $\| \mathbf{a}\|$. 
Meanwhile, $\mathbf{a}[m:m+n]$ represents the subvector of $\mathbf{a}$
made up of its $m$th to $(m+n)$-th elements.
$\log_{2}(\cdot) $ denotes the   logarithm with base 2. Furthermore,
$\mathcal{CN}(\mu,\sigma^{2})$  denotes the 
circularly symmetric complex  Gaussian (CSCG)
distribution with mean $\mu$ and variance $\sigma^{2}$.


\section{System Model}
\begin{figure}
	\centering
\includegraphics[height=2.65in, width=3.55in]{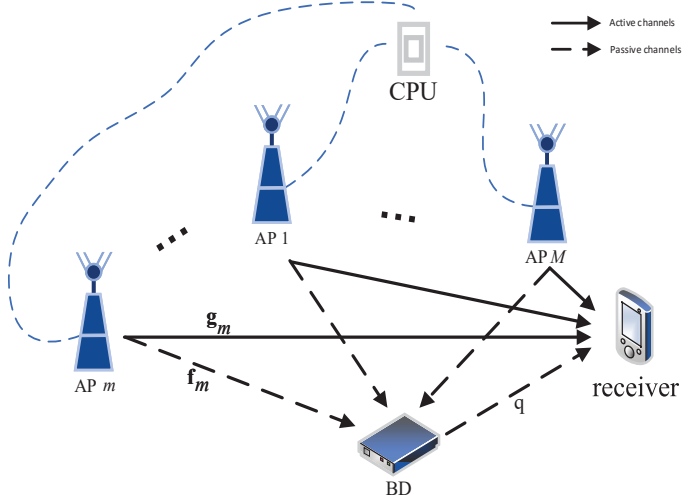}
\caption{Cell-free symbiotic radio communication system, where $M$ distributed  APs cooperatively 
transmit primary information to the receiver and concurrently  support the
 secondary passive backscattering communication.    }
\end{figure}
As shown in Fig. 1, we consider a cell-free symbiotic radio system, which 
consists of  $M$ distributed  APs, one information receiver, and one BD. 
The $M$ APs   cooperatively send primary information to the receiver,  
and simultaneously  support the  BD for secondary communication via passive
backscattering to the same  information receiver.  The considered system may model
a wide range of  applications, e.g., with the 
receiver corresponding to smartphones and the BD being the smart 
home sensor node.  
We assume that each AP is equipped with $N$ antennas, whereas the receiver and BD each has 
one antenna.  Denote by $\mathbf{g}_m\in \mathbb{C}^{N\times 1}$ and 
$\mathbf {f}_m\in \mathbb{C}^{N\times 1}$
the  multiple-input single-output (MISO) channels from the $m$th AP to the receiver and BD, 
respectively, where $m=1,...,M$. Further denote by $q\in \mathbb{C}$  the channel coefficient 
from the BD to the receiver. Thus, the  cascaded backscatter channel
from the $m$th AP to the receiver via the BD is $q\mathbf{f}_{m}$.

In this paper, we focus on the PSR setup\cite{RLong2020a},
where the symbol durations of the primary and secondary signals are equal. Let
$s(n)\sim\mathcal{CN}(0,1)$ 
and $c(n)\sim\mathcal{CN}(0,1)$   denote the CSCG
information-bearing symbols of the primary 
and secondary signals, respectively. 
Further denote by   $\mathbf{w}_{m} \in 
\mathbb{C}^{N\times 1}$   the transmit beamforming vector 
of the $m$th AP, where its  power is $\|\mathbf{w}_m\|^{2}\leq  P_{m}$, 
with $P_{m}$ denoting the maximum allowable transmit power of the $m$th AP. 
The received signal by the receiver  is 
  \begin{equation}\label{eq:rn1}
    r(n)= \!\sum_{m=1}^{M} \!\big[ \mathbf{g}_{m}^{H}\mathbf{w}_{m}s(n)\!+\!
    \sqrt{\alpha}  q\mathbf{f}_{m}^{H}\mathbf{w}_{m}s(n)c(n) \big] \!+\!z(n), 
  \end{equation}
where $ \alpha $ denotes the power reflection coefficient of the BD, 
$z(n)\sim\mathcal{CN}(0,\sigma^{2})$ is the additive 
white Gaussian noise (AWGN). 
Based on the received signal $r(n)$ in (\ref{eq:rn1}), the receiver needs to decode both the 
primary and secondary signals. Since the backscatter 
link is typically much weaker than the direct link, the receiver may first decode 
the primary symbols $s(n)$, by treating the backscatter interfering  signals  as 
noise, whose power is
$\mathbb{E}[    |\sqrt{\alpha}\sum_{m=1}^{M} q \mathbf{f}_{m}^{H}\mathbf{w}_{m} s(n)c(n)|^{2}]
=   \alpha|q|^{2}|\sum_{m=1}^{M}   \mathbf{f}_{m}^{H}\mathbf{w}_{m}|^{2}$.
Therefore, the signal-to-interference-plus-noise ratio (SINR) for 
decoding the primary information is  
\begin{equation}\label{eq:gammas}
  \gamma_{s}=\frac{ |\sum_{m=1}^{M} \mathbf{g}_{m}^{H}\mathbf{w}_{m}|^{2}}{ \alpha |q|^{2}| \sum_{m=1}^{M} \mathbf{f}_{m}^{H}\mathbf{w}_{m}|^{2}+\sigma^{2}}.
\end{equation}

Note that due to the product of $c(n)$ and $s(n)$ in the second term of (\ref{eq:rn1}), 
the resulting noise for decoding $s(n)$  no longer follows Gaussian distribution.
However, 
by using the fact that for any given noise power, Gaussian noise results 
in the maximum entropy and hence constitutes the worst-case noise 
\cite{HassibiB2003a, NeeserFD1993a},  
the achievable  rate of the primary communication in (\ref{eq:rn1}) is  
\begin{equation}
  R_{s}=\log_{2}(1+\gamma_{s}).
\end{equation}

After decoding the primary information, the first term in (\ref{eq:rn1}) can be subtracted 
from the received
signal before decoding the secondary symbols $c(n)$. The resulting signal is
\begin{equation} \label{eq:4rc}
  \hat{r}_{c}(n)=\sqrt{\alpha} q \sum\nolimits_{m=1}^{M} 
  \mathbf{f}_{m}^{H}\mathbf{w}_{m} s(n)c(n)  +z(n).
\end{equation}

Note that since $s(n)$ varies across different secondary symbols $c(n)$, 
(\ref{eq:4rc}) can be interpreted as  a fast-fading
channel,  whose instantaneous 
channel gain depends on $|s(n)|^{2}$ \cite{ZhangQ2019a}. With $s(n)\sim \mathcal{CN}(0,1)$,
its squared envelope follows an exponential distribution. 
Therefore, the ergodic rate of the backscattering communication (\ref{eq:4rc}) can 
be expressed as \cite{RLong2020a, TseD2005a}
\begin{equation}\label{eq:Rc1}
  \begin{aligned}
    R_{c}  &=\mathbb{E}_{s(n)} \Big[ \log_{2}\big(1+\frac{\alpha|q|^{2}| \sum_{m=1}^{M} 
    \mathbf{f}_{m}^{H}\mathbf{w}_{m}|^{2}|s(n)|^{2} }{\sigma^{2}}\big)\Big]
    \\ &=\int_{0}^{\infty}\log_{2}(1+\beta_{c} x)e^{-x}dx
    \\ & =-e^{\frac{1}{\beta_{c}}}\mathrm{Ei}(-\frac{1}{\beta_{c}})\log_{2}e,
\end{aligned}
\end{equation}
where $\mathrm{Ei}(-\frac{1}{\beta_{c}}) $
is the exponential integral, and 
$\beta_{c}=\frac{\alpha|q|^{2}| \sum_{m=1}^{M} \mathbf{f}_{m}^{H}\mathbf{w}_{m}|^{2}}{\sigma^{2}}$ 
is the average received signal-to-noise ratio (SNR) of the backscatter link.

\section{ Rate-Region Characterization  with Perfect CSI}
In this section, under the assumption that perfect CSI is available  at the APs, 
we aim to  characterize the achievable rate-region of the active primary communication 
and passive secondary communication,  by optimizing the transmit beamforming 
of the $M$ APs. To this end, the beamforming optimization problem 
is formulated to 
maximize the ergodic rate of the backscattering communication in (\ref{eq:Rc1}), 
subject to a given targeting communication rate constraint 
for the primary communication. By varying the targeting communication rate, 
the complete Pareto boundary of the achievable communication rate-region can be obtained.  
The problem can be formulated as

\begin{subequations}\label{eq:OP1}
  \begin{align}
  \begin{split}
    \max   \limits_{ {\mathbf{w}_m, m=1,...,M}} \quad &
        R_{c}
  \end{split}\\
  \begin{split}
    {\rm s.t.} \quad &R_{s}  \geq R_{\mathrm{th}}\hfill ,
  \end{split}\\
  \begin{split}
    \quad  & \big \| {\mathbf{w}_m} \big \|^{2} \leq  P_{m},  \qquad  m=1,...,M,
  \end{split}
\end{align}
\end{subequations}
where $R_{\mathrm{th}}$ denotes  the given targeting threshold for 
the primary communication rate, and (\ref{eq:OP1}c) corresponds to the per-AP power constraint. 

It has been shown in \cite{RLong2020a} that the first-order derivative of the ergodic rate  $R_{c}$ 
in (\ref{eq:Rc1}) with respect to the average received SNR $\beta_{c}$
is non-negative. Therefore, $  R_{c} $ is a monotonically non-decreasing
function with respect to $\beta_{c}$. Thus, we may replace 
the objective function  of (\ref{eq:OP1}) by $\beta_{c}$. 
By further ignoring  those  
constant terms,  problem (\ref{eq:OP1}) can be equivalently written as


\begin{subequations}\label{eq:OP2}
  \begin{align}
  \begin{split}
        \max   \limits_{ {\mathbf{w} }} \quad &
        | 
        \mathbf{f}^{H} {\mathbf{w}}|^{2}  
  \end{split}\\
  \begin{split}
    {\rm s.t.}  \quad & \log_{2}\Big(1+ \frac{|  \mathbf{g}^{H}  {\mathbf{w}}|^{2}}
  {  \alpha |q|^{2}
|\mathbf{f}^{H}  {\mathbf{w}} |^{2} \!+\! \sigma^{2}}\Big) \geq R_{\mathrm{th}},
  \end{split}\\
  \begin{split}
    & \big\|\mathbf{w}[(m-1)N+1:mN] \big \| ^{2} \leq  P_{m}, \qquad  m=1,...,M,
  \end{split}
\end{align}
\end{subequations}
where we have defined the cascaded vectors as
$\mathbf{g}^{T}=[\mathbf{g}^{T}_{1},\mathbf{g}^{T}_{2},
...,\mathbf{g}^{T}_{M}]$,
$\mathbf{f}^{T}=[\mathbf{f}^{T}_{1},\mathbf{f}^{T}_{2},
...,\mathbf{f}^{T}_{M}]$
and $\mathbf{w}^{T}=[\mathbf{w}^{T}_{1},\mathbf{w}^{T}_{2},...,\mathbf{w}^{T}_{M}]$.

Before solving problem (\ref{eq:OP2}), we first study its feasibility property. 
Obviously, problem (\ref{eq:OP2}) will become infeasible if $R_{\mathrm{th}}$ is too large. 
It is not difficult to see that problem (\ref{eq:OP2}) is feasible if and only if $R_{\mathrm{th}}\leq\bar{R}_{s}$, 
  where $ \bar{R}_{s} $ is the optimal value to the following optimization problem
  \begin{subequations} \label{eq:OPBIS1}
    \begin{align}
    \begin{split}
      \max   \limits_{ {\mathbf{w} }} \quad &
      \log_{2}\Big(1+ \frac{|  \mathbf{g}^{H}  {\mathbf{w}}|^{2}}
      {  \alpha |q|^{2}
    |\mathbf{f}^{H}  {\mathbf{w}} |^{2} + \sigma^{2}}\Big)
    \end{split}\\
    \begin{split}
      {\rm s.t.} 
      \quad  & \big\|\mathbf{w}[(m-1)N+1:mN] \big \| ^{2} \leq  P_{m}, \qquad  m=1,...,M.
    \end{split}
  \end{align}
  \end{subequations}

Next, we consider solving problem (\ref{eq:OPBIS1}) to get the 
maximum achievable primary communication rate threshold $\bar{R}_{s}$
for problem (\ref{eq:OP2}) to be feasible. 
Note that the objective function of problem (\ref{eq:OPBIS1})  is 
nonconcave with respect to $\mathbf{w}$. Thus, problem (\ref{eq:OPBIS1}) cannot be 
efficiently solved directly with standard convex optimization technique. 
Fortunately, the efficient optimal solution can be 
obtained via  bisection method with convex optimization.
To this end,   
it is not difficult to  see that if $\mathbf{w}^{\star}$ is an optimal
solution to problem (\ref{eq:OPBIS1}), 
so is $\mathbf{w}^{\star}e^{j\phi}$ for any  phase rotation $\phi$. 
This is because any arbitrary phase rotation for the beamforming vector $\mathbf{w}^{\star}$ 
does not change the objective function in (\ref{eq:OPBIS1}a) nor the 
constraint in (\ref{eq:OPBIS1}b). Therefore, without loss of 
optimality to problem (\ref{eq:OPBIS1}), we may assume that
$\mathbf{g}^{H}  {\mathbf{w}}$ is a nonnegative real number, i.e.,  
 $\mathrm{Re}\{\mathbf{g}^{H}  {\mathbf{w}} \}\geq 0$, and   ${\rm Im}\{ 
\mathbf{g}^{H}  {\mathbf{w}}   \}=0$.
As a result, by further introducing a  slack variable $\mu$,
 problem (\ref{eq:OPBIS1}) can be equivalently written as

  \begin{subequations}\label{eq:OPBIS2}
    \begin{align}
    \begin{split}
      \max   \limits_{ {\mathbf{w}}} \quad &
          \mu
    \end{split}\\
    \begin{split}
      {\rm s.t.} \quad &\log_{2}\Big(1\!+ \!\frac{ \big( \mathrm{Re} \{\mathbf{g}^{H}  
      {\mathbf{w}}\}\big)^{2}}
    { \alpha |q|^{2}
     |\mathbf{f}^{H} {\mathbf{w}}  |^{2} 
     \!+\! \sigma^{2} }\Big)\! \geq \mu,
    \end{split}\\
    \begin{split}
      \quad  & \big\|\mathbf{w}[(m-1)N+1:mN] \big \| ^{2} \leq  P_{m}, \qquad  m=1,...,M ,
    \end{split}\\
    \begin{split}
      & \mathrm{Im} \{ \mathbf{g}^{H}  {\mathbf{w}}   \}=0.
     \end{split}
  \end{align}
  \end{subequations}

Furthermore, for any given $\mu$, we may formulate  the following feasibility problem

  \begin{subequations}\label{eq:OPBIS3}
    \begin{align}
    \begin{split}
      \mathrm{Find} \quad& {\mathbf{w} }  
    \end{split}\\
    \begin{split}
      {\rm s.t.} \quad & \big \|[\sigma , \sqrt{\alpha} q
      \mathbf{f}^{H} {\mathbf{w}}  ] \big \|_{2}\leq \frac{
    \mathrm{Re} \{\mathbf{g}^{H}  
      {\mathbf{w}}\}}
      {\sqrt{2^{\mu}-1}},
    \end{split}\\
    \begin{split}
      \quad  & \big\|\mathbf{w}[(m-1)N+1:mN] \big \| ^{2} \leq  P_{m}, \qquad  m=1,...,M ,
    \end{split}\\
    \begin{split}
      & \mathrm{Im} \{ \mathbf{g}^{H}  {\mathbf{w}}   \}=0.
     \end{split}
  \end{align}
  \end{subequations}

Note that (\ref{eq:OPBIS3}b) is equivalent to (\ref{eq:OPBIS2}b), 
which is expressed as an SOC constraint for any given $\mu$. 
Thus, problem (\ref{eq:OPBIS3}) is a SOC programming (SOCP)  
problem, which
can be efficiently solved by standard convex optimization technique 
or existing software tools such as CVX \cite{Grantm2008}. 
If problem (\ref{eq:OPBIS3}) is feasible, then the optimal value  $\bar{R}_{s}$ of (\ref{eq:OPBIS1}) 
satisfies $\bar{R}_{s} \geq \mu$; otherwise,  $\bar{R}_{s} < \mu$.
As a result, the optimal solution to problem (\ref{eq:OPBIS1}) can 
be obtained by solving the SOCP feasibility problem (\ref{eq:OPBIS3}), 
together with the efficient bisection method to update $\mu$,
which is summarized in Algorithm 1.

After obtaining $\bar{R}_{s}$ by optimally solving the 
optimization problem (\ref{eq:OPBIS1}) with Algorithm 1, 
we consider the optimization problem (\ref{eq:OP2}) for any given rate threshold 
$R_{\mathrm{th}}\leq\bar{R}_{s}$, so that (\ref{eq:OP2}) is guaranteed to be feasible. 
We first show that when $R_{\mathrm{th}}$ is small enough, 
the optimal solution to problem (\ref{eq:OP2}) can be obtained in closed-form, 
which is stated in the following Theorem.  

\begin{theorem}
 When $R_{\mathrm{th}}\leq\widehat{R}_{s}$, where

\begin{equation}\label{eq:Rth<widehideR}
  \widehat{R}_{s} \triangleq \log_{2}(1+\frac{ \big |\sum_{m=1}^{M}\sqrt{P_{m}} 
    { \frac{ {\mathbf{g}}^{H}_{m}   {\mathbf{f}}_{m}}
    {\Vert  {\mathbf{f}}_{m} \Vert} } \big |^{2}}{ \alpha |q|^{2} \big| \sum_{m=1}^{M}\sqrt{P_{m}}
    {\Vert {\mathbf{f}}_{m} \Vert}\big |^{2}+\sigma^{2}}),
\end{equation}  
the optimal solution and optimal objective value to problem (\ref{eq:OP2}) 
can be obtained in closed-form as 
\begin{align}
  & \mathbf{ w}_{m}^{\star}=\sqrt{P_{m}} \frac{{\mathbf{f}}_{m}}
      {\Vert {\mathbf{f}}_{m} \Vert} , \qquad m=1,...,M,\label{eq:wMRT}\\
      &| \mathbf{f}^{H} {\mathbf{w}}^{\star}|^{2}   =\big ( \sum_{m=1}^{M} \sqrt{P_{m}}
\Vert {\mathbf{f}}_{m} \Vert  \big)^{2}.\label{eq:fw}
\end{align}

  \label{thm-1}
\end{theorem}

\begin{proof}
Please refer to Appendix A. 
\end{proof}
$\hfill\blacksquare$

Under the condition of Theorem 1 and with the closed-form optimal solution 
to problem (\ref{eq:OP2}) given in (\ref{eq:wMRT}), the resulting 
primary communication rate $\widehat{R}_{s}$ is given in (\ref{eq:Rth<widehideR}).
Furthermore, by substituting (\ref{eq:fw}) into (\ref{eq:Rc1}), 
the secondary communication rate $ \widehat{R}_{c}$ is obtained 
in closed-form as
\begin{align}\label{eq:Rcwidehat}
     \widehat{R}_{c}&=-e^{\frac{1}{\widehat{\beta}_{c}}}\mathrm{Ei}(-\frac{1}{\widehat{\beta}_{c}})\log_{2}e,
\end{align}
where  $\widehat{\beta}_{c}=\frac{\alpha|q|^{2}  \big | \sum_{m=1}^{M} \sqrt{P_{m}}
{\Vert {\mathbf{f}}_{m} \Vert} \big |^{2}}{\sigma^{2}}$ is the average received SNR of the
backscatter link.

\begin{algorithm}[t]
  \caption{Optimal solution to problem (\ref{eq:OPBIS1})  }
  \label{alg:algorithm-1}
  \hspace*{0.02in}{\textbf{Input:  }}The channel coefficients $q \mathbf{f}$
  and $\mathbf{h}$, noise power $\sigma^{2}$, 
   power reflection coefficient $\alpha$, termination threshold $\kappa_{1}$, 
   maximum transmit power    $P_{m}, m=1,...,M$.\\
  \hspace*{0.02in}{\textbf{Output:  }} The optimal   solution  $\mathbf{w}^{\star}$ and
  the maximum primary communication rate $\bar{R}_{s}$ to problem (\ref{eq:OPBIS1}).
  \begin{algorithmic}[1]
      \State Initialization: $\mu_{\mathrm{min}}=0$, and $\mu_{\mathrm{max}}$ 
      to a sufficiently large value.
      \While{$\mu_{ \mathrm{max}} -\mu_{\mathrm{min} }  > \kappa_{1}{\mu_{\mathrm{min}}  }$}
        \State $\mu=\frac{\mu_{\mathrm{min}}+\mu_{\mathrm{max}}}{2}$.
        \State For the given $\mu$, solve the feasibility problem  (\ref{eq:OPBIS3}).
        \If  {problem (\ref{eq:OPBIS3}) is feasible and let $\mathbf{w}^{*}$ denote its solution,}
        \State $\mu_{\mathrm{min}}=\mu,\bar{R}_{s}=\mu, \mathbf{w}^{\star}= \mathbf{w}^{*}$.
        \Else
        \State $\mu_{\mathrm{max}}=\mu $.
        \EndIf 
      \EndWhile
        \State Output $\bar{R}_{s}$ and $\mathbf{w}^{\star}$.
  \end{algorithmic}
\end{algorithm}

With the above discussions, the remaining task for solving problem (\ref{eq:OP2}) 
is to consider the case $\widehat{R}_{s} < R_{\mathrm{th}} \leq \bar{R}_{s} $. 
In this case, due to the non-concave objective function (\ref{eq:OP2}a) and 
the nonconvex constraint (\ref{eq:OP2}b), problem (\ref{eq:OP2}) is non-convex. 
Thus, it is difficult to find the optimal solution efficiently. 
Fortunately, an efficient Karush–Kuhn–Tucker (KKT) local optimal solution can be obtained by 
using the SCA technique. Towards this end, it is first observed that 
similar to (\ref{eq:OPBIS3}b), without loss of optimality, 
the rate constraint in (\ref{eq:OP2}b) can be written as a SOC constraint. 
Thus, problem (\ref{eq:OP2}) can be equivalently written as

\begin{subequations}\label{eq:OPSCA2}
  \begin{align}
  \begin{split}
        \max   \limits_{\mathbf{w}} \quad & 
        |  \mathbf{f}^{H} {\mathbf{w}} |^{2}
  \end{split}\\
\begin{split}
  {\rm s.t.}  \quad & \big \|[\sigma, \sqrt{\alpha} q
  \mathbf{f}^{H} {\mathbf{w}}  ] \big \|_{2}\leq \frac{
\mathrm{Re} \{\mathbf{g}^{H}  
  {\mathbf{w}}\}}
  {\sqrt{2^{R_{\mathrm{th}}}-1}},
\end{split}\\
  \begin{split}
    & \big\|\mathbf{w}[(m-1)N+1:mN]\big\| ^{2} \leq  P_{m}, \qquad  m=1,...,M,  
  \end{split}\\
  \begin{split}
    &  {\rm Im}\{\mathbf{g}^{H}  {\mathbf{w}}   \}=0.
  \end{split}
\end{align}
\end{subequations}

Problem (\ref{eq:OPSCA2}) is still non-convex as the objective function   
(\ref{eq:OPSCA2}a) is a convex  function with respect to $\mathbf{w}$, 
the maximization of which is a non-convex optimization problem. 
To address this issue, the SCA technique is 
applied to find a KKT local optimal solution iteratively \cite{zeng2017energy,Zappone2017a,Marks1978a}.
Specifically,  consider the current iteration $l$, in 
which the local point $\{\mathbf{w}^{(l)} \} $ is obtained in the previous iteration.
Define $F(\mathbf{w})=
| \mathbf{f}^{H}  {\mathbf{w}}|^{2}$, which is a convex differentiable function
with respect to $\mathbf{w}$. 
By using the fact that the first-order Taylor expansion of a
convex differentiable function provides a global lower bound 
\cite{hjorungnes2011complex,Boyds2004a}, we have 

  \begin{equation}\label{eq:SCA}
    \begin{aligned}
    F ({\mathbf{w}} )
    & \geq F(\mathbf{w}^{(l)})+ 2{\rm Re} \left\{
      {{\mathbf{w}}^{(l)}}^{H} 
          \mathbf{f} 
          \mathbf{f}^{H} 
          (
      \mathbf{w}-
      \mathbf{w}^{(l)})
    \right\}
    \\ & \triangleq  F_{ {\rm low} }({\mathbf{w}}|{\mathbf{w}^{(l)}}), \forall \mathbf{w}.
  \end{aligned}
  \end{equation}

Therefore, by replacing the objective function in (\ref{eq:OPSCA2}a) with its global
lower bound in (\ref{eq:SCA}), 
we have
the following optimization problem
\begin{subequations}\label{eq:OPSCA3}
  \begin{align}
  \begin{split}
        \max   \limits_{ {\mathbf{w}}} \quad &
        | \mathbf{f}^{H}  {\mathbf{w}^{(l)}}|^{2} + 2{\rm Re} \left\{
      {{\mathbf{w}}^{(l)}}^{H} 
          \mathbf{f} 
          \mathbf{f}^{H} 
          (
      \mathbf{w}-
      \mathbf{w}^{(l)})
    \right\}
  \end{split}\\
  \begin{split}
    {\rm s.t.}  \quad & \big \|[\sigma, \sqrt{\alpha}q
   \mathbf{f}^{H} {\mathbf{w}}  ] \big \|_{2}\leq \frac{
  \mathrm{Re} \{\mathbf{g}^{H}  
    {\mathbf{w}}\}}
    {\sqrt{2^{R_{\mathrm{th}}}-1}},
  \end{split}\\
  \begin{split}
    &\big \|\mathbf{w}[(m-1)N+1:mN]\big\| ^{2} \leq  P_{m}, \qquad  m=1,...,M,   
   \end{split}\\
  \begin{split}
    &  {\rm Im}\{\mathbf{g}^{H}  {\mathbf{w}}  \}=0.
  \end{split}
\end{align}
\end{subequations}

For any given local point $\mathbf{w}^{(l)}$, 
the objective function of (\ref{eq:OPSCA3}a) is 
a concave affine function of the optimization variable $\mathbf{w}$, 
and all constraints 
are convex. Therefore, problem (\ref{eq:OPSCA3}) is a convex optimization problem, 
which can be efficiently solved with standard convex optimization techniques 
or readily available software toolboxes, such as CVX \cite{Grantm2008}.
Thanks to the global lower bound in (\ref{eq:SCA}), the optimal objective value of
the convex optimization problem (\ref{eq:OPSCA3}) provides at least a lower 
bound to that of the non-convex optimization problem (\ref{eq:OPSCA2}). 
By successively updating the local point $\mathbf{w}^{(l)}$ and 
solving (\ref{eq:OPSCA3}), a monotonically non-decreasing objective value of 
(\ref{eq:OPSCA2}) can be obtained. The algorithm is summarized in Algorithm 2.


\begin{algorithm}[t]
  \caption{ SCA for problem (\ref{eq:OPSCA2})  }
  \label{alg:algorithm-2}
  \hspace*{0.02in}{\textbf{Input:  }}The channel coefficients $q \mathbf{f} $
  and $\mathbf{h} $,   noise power $\sigma^{2}$, 
   power reflection coefficient $\alpha$,  maximum  transmit power $P_{m}, m=1,...,M,$ 
   rate threshold $R_{\mathrm{th}}$ and termination threshold $\kappa_{2}$. \\
  \hspace*{0.02in}{\textbf{Output:  }} The beamforming solution  $\mathbf{w}^{\star}$.
  \begin{algorithmic}[1]
      \State Initialization: set the iteration number $l=0$, and 
      and initialize $\mathbf{w}^{(0)}$, so that it is feasible to (\ref{eq:OPSCA2}).
      \Repeat
        \State For the given local point $\mathbf{w}^{(l)}$, 
        solve the convex optimization problem (\ref{eq:OPSCA3})
        and denote the optimal solution as $\mathbf{w}^{\star(l)}$.
        \State Update the local point with $\mathbf{w}^{(l+1)}=\mathbf{w}^{\star(l)}$.
        \State Update $l=l+1$,
      \Until the fractional increase of the objective value of (\ref{eq:OPSCA2}) 
      is below the threshold $\kappa_{2}$.
  \end{algorithmic}
\end{algorithm}

Let  $F(\mathbf{w}^{\star(l)})=|  \mathbf{f}^{H} {\mathbf{w}^{\star(l)}} |^{2}$
denote the objective value of problem (\ref{eq:OPSCA2}) with the beamforming vector
obtained during the $(l)$-th iteration of Algorithm 2. We have the following
Lemma:
\begin{lemma} 
  \label{lem-1}
  \normalfont  The value $F(\mathbf{w}^{\star(l)})$ obtained 
  during each iteration of Algorithm 2 is monotonically non-decreasing, i.e.,
  $F(\mathbf{w}^{\star(l+1)}) \geq F(\mathbf{w}^{\star(l)}), \forall l$.
  Besides, the sequence $\{ \mathbf{w}^{\star(l)}  \}, l=1,2,\cdots,$ 
  converges to a KKT solution of the 
  original non-convex problem (\ref{eq:OPSCA2}).
  \end{lemma}

\begin{proof}
  Please refer to Appendix B.
\end{proof}  
$\hfill\blacksquare$


\section{Channel Estimation and Rate-Region Characterization  with Imperfect CSI}

Note that the above analysis is based on the assumption of perfect CSI on 
$\mathbf{f}_{m}$, $\mathbf{g}_{m},m=1,...,M$, and $q$. In practical wireless 
communication systems, these channels need to be acquired via e.g., 
pilot-based channel estimation. 
In the following, we propose an efficient   channel estimation method for cell-free symbiotic 
radio systems based on two-phase uplink training. Furthermore,  the achievable rates
taking into account the channel estimation errors are derived, and 
the beamforming optimization problem is revisited  with imperfect CSI
to characterize the achievable rate-region with imperfect CSI.

In the first phase of the proposed channel estimation method, 
pilot  symbols are sent by the receiver while the BD is muted, so as to estimate
the direct-link channels $\mathbf{g}_{m}, m=1,...,M$. In the second phase, pilots are 
sent both by the receiver and the BD, so that, together with the estimation of the direct-link 
channels $\mathbf{g}_{m}$, the cascaded backscatter channels $q\mathbf {f}_{m}$, are 
estimated. The details are elaborated in the following.

\subsection{Direct-Link Channel Estimation  }

First, we discuss the uplink training-based estimation of the direct-link channels between 
the receiver and the $M$ APs. Denote by $\tau_{1}$ the length of the uplink training sequence, and
let $P_{t}$ be the training power. Further denote by $\boldsymbol{\varphi}_{1}
\in \mathbb{C}^{\tau_{1}\times 1}$ the pilot sequence, where $\|\boldsymbol{\varphi}_{1}\|^{2}= \tau_{1}$. 
The received training signals by   $N$ antennas of the $m$th AP over the $\tau_{1}$ symbol 
durations, which is denoted as $\mathbf{Y}'_{m}\in  \mathbb{C}^{ N \times \tau_{1}} $, can be written as
\begin{equation}
  \mathbf{Y}'_{m}=\sqrt{P_{t}}  \mathbf{g}_{m} \boldsymbol{\varphi}_{1}^{H} 
  +\mathbf{Z}'_{m},
  \qquad m=1,...,M,
\end{equation}
where $\mathbf{Z}'_{m}$ denotes the i.i.d CSCG noise with zero-mean and power $\sigma^2$. 
With the pilot sequence $\boldsymbol{\varphi}_{1}$ known at the APs,  
$\mathbf{Y}'_{m}$ can be projected to $\boldsymbol{\varphi}_{1}$, which gives  
\begin{equation}
  \check{\mathbf{y}}'_{m}=\frac{1}{\sqrt{P_{t}}} \mathbf{Y}'_{m} \boldsymbol{\varphi}_{1}= \tau_{1} \mathbf{g}_{m}
   +\frac{1}{\sqrt{P_{t}}} \hat{{\mathbf{z}}}'_{m},  
\end{equation}
where $\hat{\mathbf{z}}^{'}_{m}\triangleq\mathbf {Z}^{'}_{m} \boldsymbol{\varphi}_{1}$ 
is the resulting noise vector.  It can be shown that $\hat{\mathbf{z} }^{'}_{m}$
is i.i.d. CSCG noise with power $\tau_{1}\sigma^2$, i.e., $\hat {\mathbf {z}}'_{m} \sim 
\mathcal{CN}(\mathbf{0} , \tau_{1}\sigma^{2} \mathbf I_N)$.

With $\mathbf g_{m}$ being a zero-mean random vector, its linear minimum mean square 
error estimation (LMMSE),
denoted by $\hat{\mathbf{g}}_{m}\in  \mathbb{C}^{N \times 1}$,  is \cite{KayS1993a}
\begin{equation}\label{eq:ghat1}
  \begin{aligned}
    \hat{\mathbf{g}}_{m}  &=\mathbb{E}[ \mathbf{g}_{m} \check{\mathbf{y}}^{'H}_{m} ] 
    \big(\mathbb{E}[\check{\mathbf{y}}'_{m} \check{\mathbf{y}}^{'H}_{m}]\big)^{-1}  \check{\mathbf{y}}'_{m} 
    \\ &=    \mathbf{R}_{\mathbf{g},m} (  \tau_{1}   \mathbf{R}_{\mathbf{g},m}+
    \frac{\sigma^2}{P_{t} }\mathbf{I}_{N})^{-1}\check{\mathbf{y}}_{m},
  \end{aligned} 
\end{equation} 
where $\mathbf{R}_{\mathbf{g},m}=\mathbb{E}[\mathbf{g}_{m}\mathbf{g}_{m}^{H}]$
denotes the   covariance matrix of $\mathbf g_{m}$. 
By further  decomposing the direct-link channel as $\mathbf g_{m}=\sqrt{b_{m}}\mathbf d_{m}$, 
with $b_{m}$ denoting the large-scale channel coefficient, and $\mathbf d_{m}\in \mathbb{C}^{N\times 1}$ denoting 
the zero-mean CSCG small-scale fading component, i.e., 
$ \mathbf{d}_m \sim \mathcal{CN}(\mathbf{0}, \mathbf I_N)$. Then $\mathbf g_{m}$ is 
CSCG distributed with covariance matrix 
$\mathbf{R}_{\mathbf{g},m}=b_{m} \mathbf I_{N}$, and thus LMMSE estimation is
also the optimal MMSE estimation. In this case, (\ref{eq:ghat1}) can be simplified as
\begin{equation}
  \hat{\mathbf{g}}_{m}=\frac{P_{t}b_{m}}{P_{t} \tau_{1} b_{m}+\sigma^2}\check{\mathbf{y}}'_{m}.  
\end{equation}

It can be shown that $\hat{\mathbf g}_{m}$ follows the distribution
\begin{equation}
  \hat{\mathbf{g}}_{m} \sim \mathcal{CN}(\mathbf{0},\frac{e_{1} b_{m}^{2}}{1+e_{1}  b_{m}}\mathbf{I}_{N}),
\end{equation}
where we have defined the transmit  training energy-to-noise ratio (ENR) as 
$  e_{1}\triangleq \frac{P_t\tau_1}{\sigma^2}$.

Let $\tilde{\mathbf{g}}_{m}$  denote the channel estimation error of the $m$th AP, i.e.,
$\tilde{\mathbf{g}}_{m}=\mathbf{g}_{m}-\hat{\mathbf{g}}_{m}$. 
With MMSE estimation, it is known that $\tilde{\mathbf{g}}_{m}$ is uncorrelated with 
$\hat{\mathbf{g}}_{m}$ \cite{KayS1993a}, which follows the distribution
\begin{equation}\label{eq:gtilde}
  \tilde{\mathbf{g}}_{m} \sim \mathcal{CN}(\mathbf{0},\frac{ b_{m}}{1+e_{1} b_{m}}\mathbf{I}_{N}).
\end{equation}

It is observed from (\ref{eq:gtilde}) that as the transmit training ENR $e_{1}$ increases, the variance of 
the channel estimation error reduces, as expected.

\subsection{Backscatter Channel Estimation  }
With the estimation $\hat{\mathbf g}_{m}$ for the direct-link channels
obtained in the first phase, in  the second phase, 
pilot symbols are sent from both the receiver and the
BD to estimate the cascaded passive backscatter channels $q\mathbf{f}_{m}$, $m=1,...,M$. 
Let $\tau_{2}$ denote the length of the training sequence in the second phase and 
$\boldsymbol{\varphi}_{2} \in  \mathbb{C}^{\tau_{2}\times 1}$ be the pilot sequence
sent by the  receiver with $\|\boldsymbol{\varphi}_{2}\|^2=\tau_{2}$. The
received training signal by the $m$th AP can be written as
\begin{equation}\label{eq:Ystar}
  \begin{aligned}
    \mathbf{Y}^{\star}_{m}& = \sqrt{P_{t} \alpha}q  \mathbf{f}_{m} \boldsymbol{\varphi}_{2}^{H}+
                \sqrt{P_{t} } (\mathbf{\hat{g}}_{m}+\mathbf{\tilde{g}}_{m}) \boldsymbol{\varphi}_{2}^{H} 
                +\mathbf{Z}''_{m},
  \end{aligned}
\end{equation}
where $\mathbf{Z}''_{m}$ denotes the i.i.d. CSCG noise with power $\sigma^2$. Note that without loss 
of generality, we assume that the pilot symbols backscattered by the BD are all equal to $1$. 
After subtracting  the terms related to the estimation  $ \hat{\mathbf{g}}_{m}$ 
of the direct-link channels from (\ref{eq:Ystar}), we have
\begin{equation}
  \begin{aligned}
    \mathbf{Y}''_{m}& = \sqrt{P_{t} \alpha}q  \mathbf{f}_{m} \boldsymbol{\varphi}_{2}^{H}+
                \sqrt{P_{t} } \tilde{\mathbf{g}}_{m} \boldsymbol{\varphi}_{2}^{H} 
                +\mathbf{Z}''_{m}.
  \end{aligned}
\end{equation}

With $\boldsymbol{\varphi}_{2}$ known at the APs,  the projection of $\mathbf{Y}''_{m}$ 
after scaling by $\frac{1}{\sqrt{P_t \alpha}}$,  is
\begin{equation}\label{eq:y..m}
  \check{\mathbf{y}}''_{m}=\frac{1}{\sqrt{P_{t} \alpha}}\mathbf{Y}''_{m} \boldsymbol{\varphi}_{2}
  =     \tau_{2}   \mathbf{h}_{m} + \frac{\tau_{2}}{\sqrt{ \alpha}} \tilde{\mathbf{g}}_{m}
    +\frac{1}{\sqrt{P_{t}\alpha} } \hat{{\mathbf{z}}}''_{m} ,
\end{equation}
where  the cascaded backscatter channel is defined  as $\mathbf{h}_{m}=
q  \mathbf{f}_{m}  $, and $\hat{{\mathbf{z}}}''_{m}\triangleq\mathbf {Z}''_{m} \boldsymbol{\varphi}_{2}$. 
It can be shown that $\hat{\mathbf{z} }''_{m} \sim 
\mathcal{CN}(\mathbf{0} , \tau_{2}\sigma^2 \mathbf I_N)$. 

Let $\mathbf{R}_{\mathbf{h},m}=\mathbb{E}[ \mathbf{h}_{m} \mathbf{h}_{m}^{H}]
$ denote the covariance matrix of the cascaded backscatter channel $\mathbf {h}_{m}$. 
Then the LMMSE estimation of $\mathbf {h}_{m}$ based on (\ref{eq:y..m}) is 
\begin{equation}\label{eq:hhat1}
  \begin{aligned}
    \hat{\mathbf{h}}_{m}  &=\mathbb{E}[     \mathbf{h}_{m}  \check{\mathbf{y}}^{''H} _{m} ] 
    \big(\mathbb{E}[\check{\mathbf{y}}''_{m} \check{\mathbf{y}}^{''H}_{m}]\big)^{-1}  \check{\mathbf{y}}''_{m} 
    \\ &=     \mathbf{R}_{\mathbf{h},m} 
    ( \tau_{2}  \mathbf{R}_{\mathbf{h},m} + \frac{\tau_{2}}{  \alpha } 
    \mathbf{R}_{\tilde{\mathbf{g}},{m}}+
    \frac{\sigma^{2}}{P_{t}\alpha} \mathbf{I}_{N})^{-1}\check{\mathbf{y}}''_{m},
  \end{aligned} 
\end{equation} 
where $\mathbf{R}_{\tilde{\mathbf{g}},{m}}=\mathbb{E} [\mathbf{\tilde{g}}_{m}
\mathbf {\tilde{g}}_{m}^{H}]$ is the covariance matrix of $\tilde {\mathbf {g}}_{m}$.     

If the channel coefficients in $\mathbf {f}_{m}$ are i.i.d.  distributed with 
variance $\zeta_{m}$, we then have 
$\mathbf R_{\mathbf{h},m}=\mathbb{E}[|q|^2\mathbf f_{m} 
\mathbf f_m^H]=\epsilon_{m} \mathbf I_N$,
where $\epsilon_{m}=\mathbb{E}[|q|^2] \zeta_{m}$.

Therefore, (\ref{eq:hhat1}) can be simplified as
\begin{equation}\label{eq:hhat2}
  \hat{\mathbf{h}}_{m} =  \frac{\alpha P_{t} \epsilon_{m}}{\alpha P_{t} \tau_{2} \epsilon_{m} +
    \frac{P_{t}\tau_{2}b_{m}}{1+e_{1}b_{m}}+\sigma^2} \check{\mathbf{y}}''_{m}.
\end{equation}

Define the transmit training ENR in the second phase as $e_{2}\triangleq \frac{P_{t}\tau_{2}}{\sigma^2}$.
Then with (\ref{eq:gtilde}) and (\ref{eq:hhat2}), we have
\begin{equation}
  \begin{aligned}
    \mathbf{R}_{\hat{\mathbf{h}},m} &=\mathbb{E}[ \hat{\mathbf{h}}_{m}  \hat{\mathbf{h}}_{m}^{H}]
     =\frac{\alpha e_{2}   \epsilon_{m}^{2}} { \alpha e_{2}  \epsilon_{m}
    +\frac{e_{2}b_{m}}{1+e_{1}b_{m}}+1}\mathbf{I}_{N}.
  \end{aligned}
\end{equation}

Let $\tilde{\mathbf h}_{m}=\mathbf{h}_{m}-\hat {\mathbf h}_{m}$ 
denote the estimation error. We  have
\begin{equation}\label{eq:Rhtilde}
  \begin{aligned}
    \mathbf{R}_{\tilde{\mathbf{h}},m} &
    \triangleq \mathbb{E}\big[(\mathbf {h}_{m}-\hat {\mathbf h}_{m})(\mathbf {h}_{m}-\hat {\mathbf h}_{m})^{H}\big]
      \\ &=\mathbf{R}_{\mathbf{h},m} - \mathbf{R}_{\hat{\mathbf{h}},m}
  \\ &=\frac{\epsilon_{m} (\frac{e_{2}b_{m}}{1+e_{1}b_{m}}+1)}
  {\alpha e_{2} \epsilon_{m}
  +\frac{e_{2}b_{m}}{1+e_{1}b_{m}}+1}\mathbf{I}_{N}.
  \end{aligned} 
\end{equation}

It follows from (\ref{eq:gtilde}) that if $e_{1} \rightarrow \infty$, in which case the direct-link channel 
$\mathbf g_m$  is perfectly estimated without any error, the variance of the 
estimation error in (\ref{eq:Rhtilde}) reduces to the same form as that in (\ref{eq:gtilde}). 

\subsection{Achievable Rate Analysis}
In this subsection, we derive the  achievable  primary and secondary rates  
based on the  channel estimation $ \hat{\mathbf{g}}_{m} $ 
and $ \hat{\mathbf{h}}_{m}, m=1,...,M$, by taking into account the channel estimation errors. 
By substituting $\mathbf g_{m}=\hat {\mathbf g}_{m} + \tilde{\mathbf g}_{m}$ and 
$q\mathbf f_{m} =\hat {\mathbf h}_{m}+\tilde{\mathbf h}_{m}$ into (\ref{eq:rn1}), 
the received signal   can be written as
\begin{equation}\label{eq:rnest}
  \begin{aligned}
  r(n) &=     \sum\nolimits_{m=1}^{M} \big[(\hat{\mathbf{g}}_{m}+\tilde{\mathbf{g}}_{m})^{H}  {\mathbf{w}_{m}} s(n)+
  \\ & \sqrt{\alpha}(\hat{\mathbf{h}}_{m}+\tilde{\mathbf{h}}_{m})^{H} {\mathbf{w}_{m}} s(n)c(n) \big] +z(n). 
\end{aligned}
\end{equation}

For decoding  the primary signals $s(n)$, besides the interference from the 
backscatter symbols $c(n)$, the term caused by the channel estimation error 
$\tilde{\mathbf{g}}_{m}$ is also treated as noise \cite{YZeng2014a, HassibiB2003a}.  
Therefore, (\ref{eq:rnest}) can be  decomposed as
\begin{equation}
  r_{s}(n) = {\rm DS}'\cdot s(n)+ {\rm ER}+{\rm ST}+z(n),
\end{equation}
where ${\rm DS}', {\rm ER}$, and ${\rm ST}$ denote the desired signal, 
estimation errors and the secondary transmission signal, respectively, which are given by
\begin{align}
  {\rm DS}'&= \sum\nolimits_{m=1}^{M}  \hat{\mathbf{g}}_{m} ^{H}  {\mathbf{w}_{m}} ,\label{eq:DS}\\
  {\rm ER}&= \sum\nolimits_{m=1}^{M} \big( \tilde{\mathbf{g}}_{m}  
  +\sqrt{\alpha}\tilde{\mathbf{h}}_{m}c(n)\big)^{H} {\mathbf{w}_{m}}s(n),\label{eq:ER}\\
  {\rm ST}&= \sum\nolimits_{m=1}^{M}  
  \sqrt{\alpha}\hat{\mathbf{h}}_{m} ^{H} {\mathbf{w}_{m}}s(n)c(n) .
\end{align}

Therefore, the resulting SINR can be expressed as (\ref{eq:SINRsest}) shown at the top
of the next page, and the achievable rate is $  R_s=\log_2(1+\gamma_s)$.
\newcounter{TempEqCnt} 
\setcounter{TempEqCnt}{\value{equation}} 
\setcounter{equation}{35} 

\begin{figure*}[ht] 
  \begin{equation}\label{eq:SINRsest}
    \begin{aligned}
    \gamma_{s} & = \frac{|{\rm DS}'|^{2}}{\mathbb {E}_{s(n),c(n)} \big[|{\rm ER} |^{2}\big]
    +\mathbb {E}_{s(n),c(n)} \big[|{\rm ST} |^{2}\big]+\sigma^{2}}
    \\ &= \frac{| \sum_{m=1}^{M}\hat{\mathbf{g}}_{m}^{H}  {\mathbf{w}_{m}}|^{2}}
    { \mathbb{E}_{s(n),c(n)}\big[|s(n)|^{2}| \sum_{m=1}^{M} \big( \tilde{\mathbf{g}}_{m}  
    \!+\!\sqrt{\alpha}  \tilde{\mathbf{h}}_{m}c(n)\big)^{H} {\mathbf{w}_{m}} |^{2}\big]  
    \! +\!\alpha\mathbb{E}_{s(n),c(n)}\big[
     | \sum_{m=1}^{M} \!\hat{\mathbf{h}}_{m} ^{H} \!{\mathbf{w}_{m}}  |^{2} 
     |s(n)|^{2}|c(n)|^{2}\big] 
     \!+\! \sigma^{2}  }
     \\ &=\frac{| \sum_{m=1}^{M}\hat{\mathbf{g}}_{m}^{H}  {\mathbf{w}_{m}}|^{2}}
    {   \sum_{m=1}^{M}\sum_{l=1}^{M} {\mathbf{w}_{m}^{H}} ( \tilde{\mathbf{g}}_{m} 
    \tilde{\mathbf{g}}_{l}^{H}\!+\!
     \alpha  \tilde{\mathbf{h}}_{m}\tilde{\mathbf{h}}_{l}^{H}){\mathbf{w}_{l}}  
    \! +\!\alpha
     | \sum_{m=1}^{M} \!\hat{\mathbf{h}}_{m} ^{H} \!{\mathbf{w}_{m}}  |^{2}  
     \!+\! \sigma^{2}  }.
  \end{aligned}
\end{equation}
\hrulefill  

\end{figure*}

Note that for any given channel estimations $\hat{\mathbf{g}}_{m} $ and 
$\hat{\mathbf{h}}_{m}$, since the channel estimation errors 
$\tilde{\mathbf{g}}_{m} $ and $\tilde{\mathbf{h}}_{m}$ are random, 
the SINR in (\ref{eq:SINRsest}) and hence its rate $R_s$ is random. 
By taking the expected achievable rate with respect to the random 
estimation errors $\tilde{\mathbf{g}}_{m} $ 
and $\tilde{\mathbf{h}}_{m}$, we have the result (\ref{eq:ERs}) shown 
at the top of the next page, 
where we have denoted the average channel estimation-error-plus-noise power as

\setcounter{TempEqCnt}{\value{equation}} 
\setcounter{equation}{36} 
\begin{figure*}[ht] 
  \begin{equation}\label{eq:ERs}
    \begin{aligned}
    \mathbb{E}[R_{s}]&=\mathbb{E}_{\tilde{\mathbf{g}}_{m}  
    , \tilde{\mathbf{h}}_{m}}\big[\log_{2}(1+\gamma_{s})\big] 
    \\ &\geq \log_{2}\Big(1+\frac{| \sum_{m=1}^{M}\hat{\mathbf{g}}_{m}^{H}  {\mathbf{w}_{m}}|^{2}}
    {  \mathbb{E}_{\tilde{\mathbf{g}}_{m}  
    , \tilde{\mathbf{h}}_{m}}\big[ \sum_{m=1}^{M}\sum_{l=1}^{M} {\mathbf{w}_{m}^{H}} ( \tilde{\mathbf{g}}_{m} 
    \tilde{\mathbf{g}}_{l}^{H}\!+\!
     \alpha  \tilde{\mathbf{h}}_{m}\tilde{\mathbf{h}}_{l}^{H}){\mathbf{w}_{l}}\big]  
    \! +\!\alpha
     | \sum_{m=1}^{M} \!\hat{\mathbf{h}}_{m} ^{H} \!{\mathbf{w}_{m}}  |^{2}  
     \!+\! \sigma^{2}  }\Big)
     \\ &=\log_{2}\Big(1+\frac{| \sum_{m=1}^{M}\hat{\mathbf{g}}_{m}^{H}  {\mathbf{w}_{m}}|^{2}}
     { \sum_{m=1}^{M} {\mathbf{w}_{m}^{H}} (\mathbf{R}_{\tilde{\mathbf{g}},m}  + 
      \alpha\mathbf{R}_{\tilde{\mathbf{h}},m}  ) {\mathbf{w}_{m}} + \sigma^{2}
      +\alpha
      |\! \sum_{m=1}^{M}\!\hat{\mathbf{h}}_{m}^{H}  {\mathbf{w}_{m}} |^{2}  }\Big)
     \\ &=\log_{2}\Big(1+\frac{| \sum_{m=1}^{M}\hat{\mathbf{g}}_{m}^{H}  {\mathbf{w}_{m}}|^{2}}
     {E+\! \alpha
   |\! \sum_{m=1}^{M}\!\hat{\mathbf{h}}_{m}^{H}  {\mathbf{w}_{m}} |^{2} }\Big) 
   \\ & \triangleq  \bar{R}_{s,\mathrm{LB}}.
  \end{aligned}
  \end{equation}
  \hrulefill  

\end{figure*} 

\begin{equation}\label{eq:E}
  E= \sum_{m=1}^{M}P_{m}\big[\frac{ b_{m}}{ 1+e_{1} b_{m} }
+\frac{\alpha \epsilon_{m} (\frac{e_{2}b_{m}}{1+e_{1}b_{m}}+1)}
{\alpha e_{2} \epsilon_{m}
+\frac{e_{2}b_{m}}{1+e_{1}b_{m}}+1}\big]+\sigma^{2}.
\end{equation}

In (\ref{eq:ERs}), the lower bound of the expected primary communication rate $\mathbb{E}[R_{s}]$
is denoted by $\bar{R}_{s,\mathrm{LB}}$. 
Note that the inequality in (\ref{eq:ERs}) follows from   Jensen's inequality, and the fact that
 $\log_2(1+C/x)$ 
is a convex function for $x>0$.

Next, we derive the achievable rate of the secondary signals $c(n)$. After  
decoding $s(n)$, the primary signals $s(n)$ can be subtracted from (\ref{eq:rnest}) based on  the 
estimated  channel $\hat{\mathbf{g}}_{m}$. The resulting signal is
\begin{equation}\label{eq:Rcest}
  \begin{aligned}
   {r}_{c}(n)&= \sum\nolimits_{m=1}^{M} \big[\sqrt{\alpha}  ( \hat{\mathbf{h}}_{m}+\tilde{\mathbf{h}}_{m})^{H}  {\mathbf{w}_{m}}s(n)c(n) 
  \\ &+   \tilde{\mathbf{g}}_{m}^{H}  {\mathbf{w}_{m}} s(n)\big]+z(n).
  \end{aligned}
\end{equation}

By treating the terms caused by the channel estimation error 
$\tilde{\mathbf{g}}_{m}$ and $\tilde{\mathbf{h}}_{m}$ as noise, 
(\ref{eq:Rcest}) can be decomposed as
\begin{equation}
  r_{c}(n)= {\rm DS}''\cdot c(n)+{\rm ER}+z(n),
\end{equation}
where  $ {\rm ER}$  accounts for  the
estimation errors given in (\ref{eq:ER}), and ${\rm DS}''$ denotes the desired signal
in the decoding of the secondary signals $c(n)$,  
which is given by
\begin{equation}
  {\rm DS}''= \sum\nolimits_{m=1}^{M}  \sqrt{\alpha}    \hat{\mathbf{h}}_{m}
  ^{H}  {\mathbf{w}_{m}}s(n) .
\end{equation}
 
The resulting   SINR is  
\begin{equation}\label{eq:SINRcest}
  \begin{aligned}
    \gamma_{c} &\!=\!\frac{|{\rm DS}''|^{2}}{\mathbb{E}_{s(n),c(n)}\big[|{\rm ER} |^{2}\big] +\sigma^{2}}
    \\ &\!=\!\frac{\alpha | \sum_{m=1}^{M}\hat{\mathbf{h}}_{m}^{H}  {\mathbf{w}_{m}}|^{2} |s(n)|^{2}}
    {\sum_{m=1}^{M}\sum_{l=1}^{M} {\mathbf{w}_{m}^{H}} ( \tilde{\mathbf{g}}_{m} 
    \tilde{\mathbf{g}}_{l}^{H}\!+\!
     \alpha  \tilde{\mathbf{h}}_{m}\tilde{\mathbf{h}}_{l}^{H}){\mathbf{w}_{l}}  \! 
     +  \!  \sigma^{2}  },
  \end{aligned}
\end{equation}
and the achievable rate is $  R_c=\log_2(1+\gamma_c)$.

Note that different from (\ref{eq:DS}), as the desired channel ${\rm DS}''$ also depends on the 
primary symbols $s(n)$, the SINR in (\ref{eq:SINRcest}) is a random variable that depends on 
both $|s(n)|^2$ and the channel estimation errors. Consider the expectation of $R_{c}$,
which is  taken with respect to both 
$|s(n)|^2$ and the channel estimation errors, we have $\mathbb{E}[R_{c}]$ 
given in (\ref{eq:ERc}) at the top of next page.
In (\ref{eq:ERc}) the inequality is obtained by applying Jensen's inequality to the 
convex function $\mathrm{log}_2(1+C/x)$, and
 $\beta^{'}_{c}=\frac{ \alpha | \sum_{m=1}^{M}
\hat{\mathbf{h}}_{m}^{H} {\mathbf{w}_{m}}|^{2}   }{E}$ represents the average
SINR for the secondary signals taking into account the channel estimation errors.
Similarly,  we denote the lower bound of expectation rate $\mathbb{E}[R_{c}] $ with $\bar{R}_{c,\mathrm{LB}}$.
\setcounter{TempEqCnt}{\value{equation}} 
\setcounter{equation}{42} 

\begin{figure*}[ht] 
  \begin{equation}\label{eq:ERc}
    \begin{aligned}
       \mathbb{E}[R_{c}] & =\mathbb{E}_{\tilde{\mathbf{g}}_{m}  
      , \tilde{\mathbf{h}}_{m}, s(n)}\big[\log_{2}(1+\gamma_{c})\big] 
      \geq \!\mathbb{E}_{s(n)}\Big[\log_{2}\big(1\!+\!\frac{\alpha | \sum_{m=1}^{M}\hat{\mathbf{h}}_{m}^{H}  
      {\mathbf{w}_{m}}|^{2}|s(n)|^{2} }
      { \mathbb{E}_{\tilde{\mathbf{g}}_{m}  
      , \tilde{\mathbf{h}}_{m} }\big[\sum\limits_{m=1}^{M}\sum\limits_{l=1}^{M} {\mathbf{w}_{m}^{H}} ( \tilde{\mathbf{g}}_{m} 
      \tilde{\mathbf{g}}_{l}^{H}\!+\!
       \alpha  \tilde{\mathbf{h}}_{m}\tilde{\mathbf{h}}_{l}^{H}){\mathbf{w}_{l}} \big]\! 
       +  \! \sigma^{2} }\big)\Big]
     \\ &=\mathbb{E}_{s(n)}\Big[\log_{2}\big(1+\frac{ \alpha | \sum_{m=1}^{M}
     \hat{\mathbf{h}}_{m}^{H} {\mathbf{w}_{m}}|^{2} |s(n)|^{2} }
     {E}\big)\Big]
     =\int_{0}^{\infty}\log_{2}(1+\beta^{'}_{c} x)e^{-x}dx
    \\ &=-e^{\frac{1}{\beta^{'}_{c}}}\mathrm{Ei}(-\frac{1}{\beta^{'}_{c}})\log_{2}e
    \\ & \triangleq  \bar{R}_{c,\mathrm{LB}}.
    \end{aligned}
  \end{equation}
\hrulefill  

\end{figure*}


\subsection{ Rate-Region Characterization with Imperfect CSI}

In this subsection, the achievable rate-region of cell-free symbiotic radio system 
is characterized by taking into account the CSI estimation errors.

Similar to the  optimization problem (\ref{eq:OP1}) 
with perfect CSI studied in Section  \RNum{3}, the  achievable rate-region
under channel estimation errors can be characterized by
maximizing the expected secondary communication rate 
$\bar{R}_{c,\mathrm{LB}} $ in (\ref{eq:ERc})
with a given targeting primary communication rate constraint $R_{\mathrm{th}}$
for $\bar{R}_{s,\mathrm{LB}} $ in (\ref{eq:ERs}).  
The optimization problem can be formulated as
\begin{subequations}\label{eq:OPest1}
  \begin{align}
  \begin{split}
    \max   \limits_{ {\mathbf{w}_m, m=1,...,M}} \quad &
        \bar{R}_{c,\mathrm{LB}}  
  \end{split}\\
  \begin{split}
    {\rm s.t.} \quad & \bar{R}_{s,\mathrm{LB}}  \geq R_{\mathrm{th}}\hfill ,
  \end{split}\\
  \begin{split}
    \quad  & \big \| {\mathbf{w}_m} \big \|^{2} \leq  P_{m},  \qquad  m=1,...,M.
  \end{split}
\end{align}
\end{subequations}

Since the  $\bar{R}_{c,\mathrm{LB}} $ is also monotonically
increasing with respect to $\beta^{'}_{c}$, 
the objective function (\ref{eq:OPest1}a) is equivalent to $\beta^{'}_{c}$. 
Therefore, problem (\ref{eq:OPest1}) can be equivalently written as

\begin{subequations}\label{eq:OPest2}
  \begin{align}
  \begin{split}
        \max   \limits_{ {\mathbf{w}_m, m=1,...,M}} \quad &
        | \sum_{m=1}^{M}
        \hat{\mathbf{h}}_{m}^{H} {\mathbf{w}_{m}}|^{2},  
  \end{split}\\
  \begin{split}
    {\rm s.t.}  \quad & \log_{2}\Big(1+ \frac{| \sum_{m=1}^{M}\hat{\mathbf{g}}_{m}^{H}  {\mathbf{w}_{m}}|^{2}}
  {E+\! \alpha
|\! \sum_{m=1}^{M}\!\hat{\mathbf{h}}_{m}^{H}  {\mathbf{w}_{m}} |^{2} }\Big) \geq R_{\mathrm{th}},
  \end{split}\\
  \begin{split}
  \quad &\big \| {\mathbf{w}_m} \big\|^{2} \leq  P_{m},  \qquad  m=1,...,M.
  \end{split}
\end{align}
\end{subequations}

Note that different from problem (\ref{eq:OP2}) that requires perfect CSI, 
problem (\ref{eq:OPest2}) only requires the estimated channel CSI 
$\hat{\mathbf{g}}_{m}^{H}$ and $\hat{\mathbf{h}}_{m}^{H}$, $m=1,...,M$, 
and the impact of CSI estimation error is reflected in $E$.  
Furthermore, since problem (\ref{eq:OPest2}) has exactly 
the same structure as (\ref{eq:OP2}), the feasibility analysis as well as 
the proposed solution in Theorem 1 and Algorithm 2 can be directly used to 
solve problem (\ref{eq:OPest2}). The details are omitted to avoid repitition.



\section{Simulation Results}

In this section,  simulation results are provided to 
evaluate the performance of our proposed design for
cell-free symbiotic radio systems. 
Without loss of generality, we establish a Cartesian coordinate system, 
where the BD is located at the 
origin (0,0), and the receiver is located at (5m, 0). 
Furthermore,   $M=16$ APs, each with $N=4$ antennas, are 
evenly spaced in a square area of size $750\mathrm{m} \times 750\mathrm{m} $, 
i.e., their locations correspond to the $4 \times 4$ grid points, with 
the x- and y-coordinates chosen from the set \{-375m, -125m, 125m, 375m\}.  
The  small-scale fading coefficients for the channels of different communication links
follow the i.i.d. CSCG distribution with zero mean and unit variance. 
Furthermore, the large-scale channel gains of AP-to-BD and AP-to-receiver links are  modeled as 
$b_{m}=\beta_{0}  d_{m}^{-\gamma}$,
where $\beta_{0}=(\frac{\lambda}{4\pi})^{2}$ is the reference channel gain, 
with $\lambda=0.0857$m denoting the wavelength, $d_{m}$ represents the 
corresponding channel link distance of AP-to-BD or AP-to-receiver, 
and $\gamma$ denotes the path loss 
exponent. We set $\gamma=2.7$ for the AP-to-BD and AP-to-receiver channels.
And the large-scale channel coefficients $\epsilon_{m}$ of the cascaded 
AP-BD-receiver channels are modeled as
$\epsilon_{m}=0.001 \zeta_{m}$, where $\zeta_{m}$ represents the large-scale coefficients
of AP-to-BD channels. 
The power reflection coefficient is $\alpha=1$, and the 
transmitter-side SNR for data and pilot transmission are set as 
$\frac{P_{m}}{\sigma^{2}}=\frac{P_{t}}{\sigma^{2}}=130$ dB, $m=1,...,M$,
which may correspond to $P_{m}=P_{t}=20$ dBm and $\sigma^2=-110$ dBm.
The termination  threshold for Algorithm 1 is set as $\kappa_{1}=0.5\%$.
We further denote  $\tau_{\mathrm{total}}=\tau_{1}+\tau_{2} $ as the total  
pilot length used for channel estimation in  phase 1 and phase 2 and  
 $l_{1}=\frac{\tau_{1}}{\tau_{\mathrm{total}}}$ and    
$l_{2}=\frac{\tau_{2}}{\tau_{\mathrm{total}}}$
denote the ratios of the total pilot length allocated for the first and 
second training phases, respectively. 

\begin{figure}[!t]
  \centering
\centerline{\includegraphics[height=3in, width=3.809in]{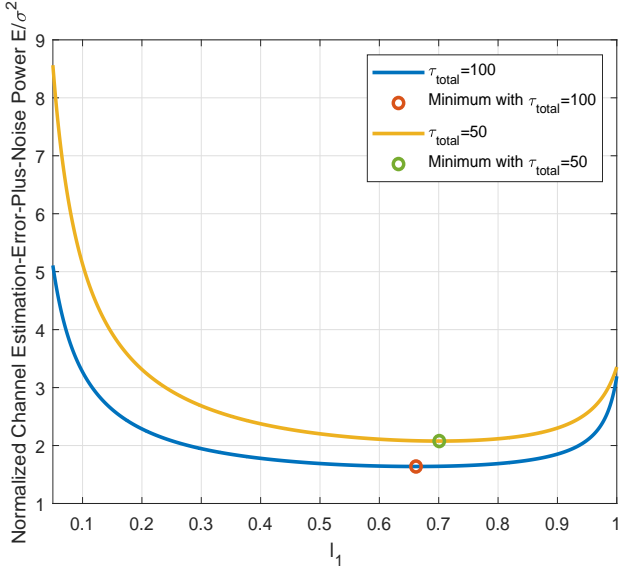}}
\caption{Normalized channel estimation-error-plus-noise power $E/\sigma^2$
versus primary pilot length ratio $l_{1}$ in phase 1.  }
\end{figure}

\begin{figure}[!t]
  \centering
\centerline{\includegraphics[height=3in, width=3.809in]{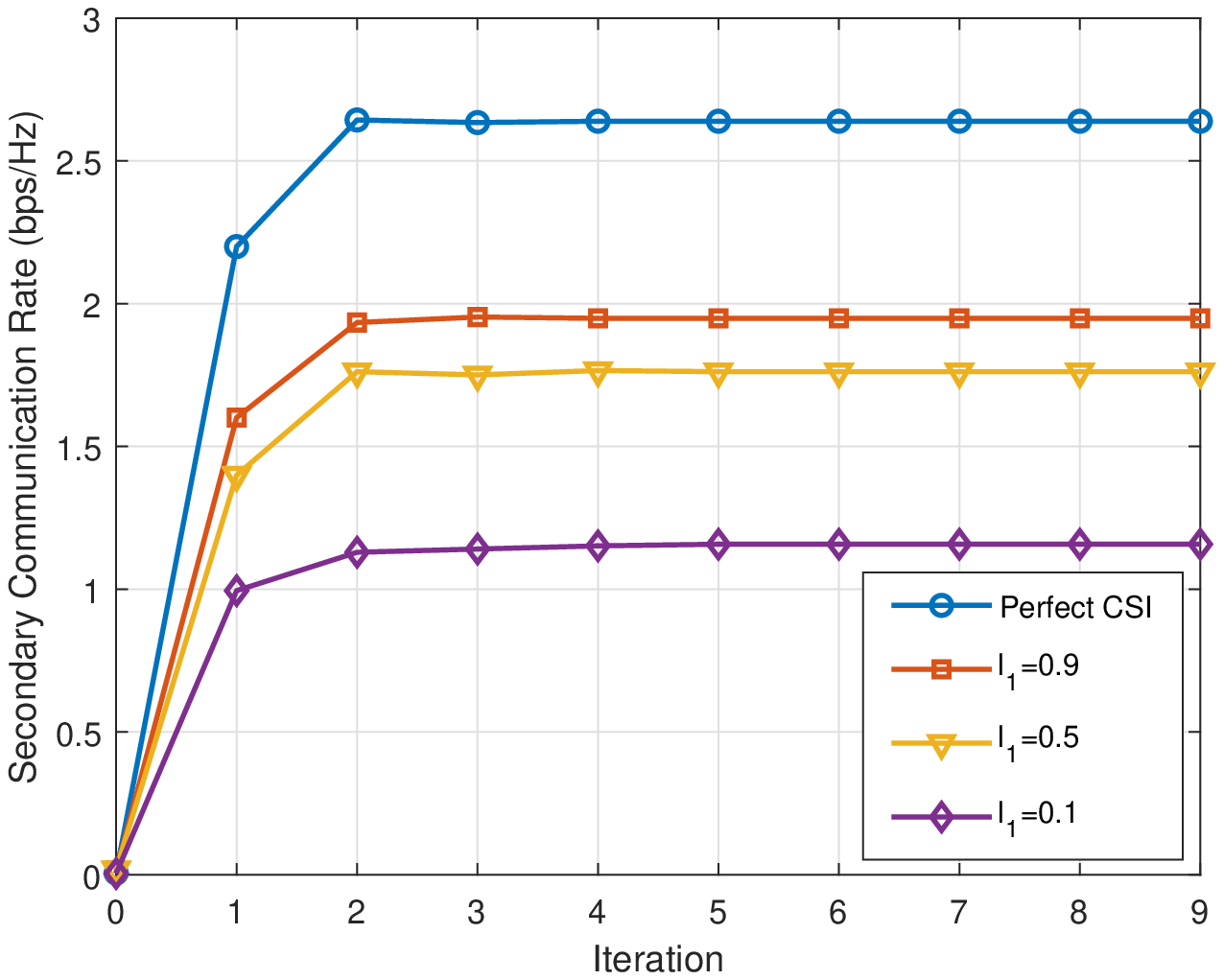}}
\caption{Convergence of the SCA iteration process with 
different primary pilot length ratio $l_{1}$, 
where $R_{\mathrm{th}} = 12 \mathrm{bps/Hz}$, $\tau_{\mathrm{total}}=50$.  }
\end{figure}
\begin{figure}[!t]
  \centering
\centerline{\includegraphics[height=3in, width=3.809in]{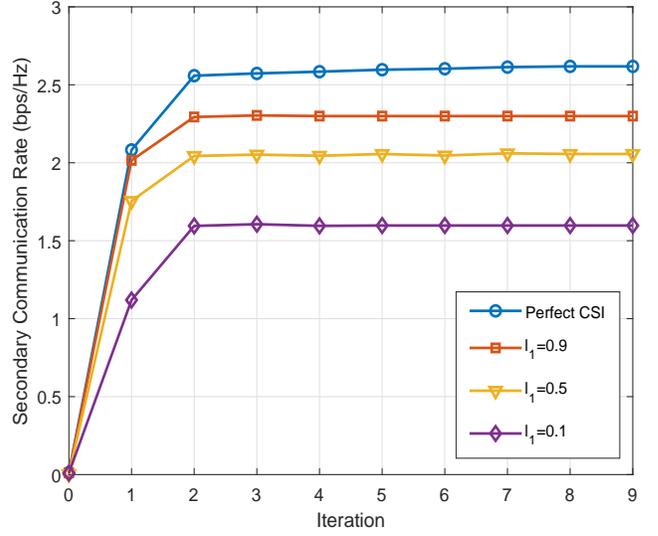}}
\caption{Convergence of the SCA iteration process with 
different primary pilot length ratio $l_{1}$, 
where $R_{\mathrm{th}} = 12 \mathrm{bps/Hz}$, $\tau_{\mathrm{total}}=100$.   }
\end{figure}

\subsection{Channel  Estimation Error    and Convergence of SCA Algorithm}

As can be inferred from (\ref{eq:ERs}) and (\ref{eq:ERc}), with imperfect CSI, not only the additive noise but also
the channel estimation error will impair the performance of cell-free symbiotic radio
systems. Therefore, Fig. 2 shows the normalized  
channel estimation-error-plus-noise power, i.e., 
$E/\sigma^2$, where $E$ is given in equation (\ref{eq:E})
and the normalization is taken with respect to the noise power $\sigma^2$.  
The ratio $l_{1}$ of the  training phase 1
varies from $0.05$ to $1$. The  total pilot length 
$\tau_{\mathrm{total}}$ is set as 50 and 100, respectively.
Note that with perfect CSI,  we have
$\frac{E}{\sigma^2}=1$.
It is observed from Fig. 2 that 
the curve with   $\tau_{\mathrm{total}}=100$  always stays below
the curve with $\tau_{\mathrm{total}}=50$.
This is expected since  the accuracy of channel estimation
can be improved with a longer  pilot length.
Furthermore, for both  $\tau_{\mathrm{total}}=100$ and 
$\tau_{\mathrm{total}}=50$, as ratio $l_{1}$ increases, $E/\sigma^2$ first decreases, and 
then increases slowly after reaching their lowest point. This is expected 
since for a given total training length, there exists a trade-off between the 
estimation error  in phase 1 and phase 2, as can be inferred from (\ref{eq:E}). 
Specifically, the increase of $l_{1}$ increases the training ENR $e_{1}$, 
which  effectively reduces the primary channel estimation error
in (\ref{eq:gtilde}),   while
the secondary cascaded channels suffer from more serious estimation error.
Besides, it can be inferred from Fig. 2 that the 
average channel estimation-error-plus-noise power $E$ is more sensitive to the pilot 
length of phase 1 than that of phase 2.
Specifically, when a short pilot was allocated to phase 1 (say $l_{1}=0.1$),
the corresponding  $E/\sigma^2$ is much higher than its counterpart
with $l_{1}=0.9$, which corresponds a low pilot allocation to phase 2.
This   is expected since
the estimation error of the primary channels in phase 1 also affects 
the channel estimation accuracy of the secondary cascaded link in phase 2,
as can be inferred from (\ref{eq:Rhtilde}).
As a result, with a limited pilot length, higher priority should be given to 
  phase 1 in order to decrease the estimation error of the whole system.
Fig. 2 also labels the optimal pilot allocation 
for which ${E/\sigma^2}$ achieves the minimum value. It is observed that 
around the optimal point, 
the curves are rather flat,     
indicating that the impact of the channel estimation error 
would be comparable for a wide range of pilot length allocations, say for $0.3<l_{1}<0.9$. 
Therefore, in the following,  we choose $l_{1}=0.1,0.5,0.9$ 
as the representative values to show the impact of different pilot length 
allocations on system performance.

Fig. 3 and Fig. 4  show the convergence of 
the proposed SCA based algorithm in Algorithm 2  with 
total pilot lengths of $\tau_{\mathrm{total}}= 50$ and $100$, respectively. The iterations start with
randomly generated initial local points, with a rate threshold of
$R_{\mathrm{th}} = 12$ bps/Hz and an iteration
terminating threshold of $\kappa_{2}=0.5\%$. Note that different curves in 
Fig. 3 and Fig. 4 show the convergence with perfect CSI and different primary pilot
length ratios $l_{1}$, respectively.   
As shown in the figure, the passive secondary communication rate, 
which is the optimization objective of problem (\ref{eq:OP2}), 
increases monotonically  during iterations, 
which is in accorance with Lemma 1. Furthermore, Fig. 3 and Fig. 4 show
that only a few iterations are needed for Algorithm 2 to converge.

\subsection{Achievable Rate-Region  }

Fig. 5 shows the achievable rate-region of the primary and secondary communication rates 
with perfect and imperfect CSI, respectively. In the case of imperfect CSI,
the primary pilot length ratios are selected as $l_1=0.1,0.5,0.9$, and
the simulation results are obtained by running 500 experiments 
at each primary pilot length ratio based on the  random cell-free symbiotic radio
system channel realization. 
The total pilot length $\tau_{\mathrm{total}}$ is set to 50.
Note that each point
of the curve corresponds to a primary-secondary rate pair, by varying the primary communication 
threshold $R_{\mathrm{th}}$ from $\widehat{R}_{s}$ to $\bar{R}_{s}$ with 
step size 1. 
Meanwhile,  the dot-dash line shows the portion with 
$R_{\mathrm{th}}\leq \widehat{R}_{s} $, 
where the optimization problem has a closed-form solution given in Theorem 1, and the
primary and secondary communication rate 
can be obtained  in (\ref{eq:Rth<widehideR}) and (\ref{eq:Rcwidehat}).
If $ \widehat{R}_{s} < R_{\mathrm{th}} \leq \bar{R}_{s} $, 
the achievable rate of the secondary backscattering communication
decreases  monotonically as the primary communication  threshold $R_{\mathrm{th}}$ increases, 
and eventually approaching zero.
It is observed from
Fig. 5 that with the total pilot length $\tau_{\mathrm{total}}$ fixed, 
the achievable rate-region taking into account imperfect CSI estimation,  
is highly dependent on the allocation of the pilot $l_{1}$. 
Out of the three pilot allocation considered,  $l_{1}=0.5$ gives the best performance.
This is consistent with Fig. 2, where $l_{1}=0.5$ gives 
smaller channel estimation error than $l_{1}=0.1$ and $l_{1}=0.9$.
Furthermore, Fig. 5 also shows that $l_{1}=0.9$ outperforms $l_{1}=0.1$. 
This is also consistent with Fig. 2, by comparing their 
respective channel estimation errors.

Fig. 6 shows the achievable rate-region with a fixed 
total pilot length $\tau_{\mathrm{total}}=100$.
Similar to Fig. 5, a higher 
primary communication  threshold $R_{\mathrm{th}}$ brings a lower secondary communication 
rate. Besides, when comparing Fig. 6 to Fig. 5, it is observed 
that larger rate-region is achieved for any given pilot allocation $l_{1}$
in Fig. 6  than its counterpart in Fig. 5. 
This is expected since a longer pilot length enhances  the estimation 
accuracy of both the active primary and passive  secondary channels, 
thus enlarging the achievable rate-region  significantly.
Meanwhile, it is observed that 
the performance gap between $l_{1}=0.5$ and $l_{1}=0.9$ 
is smaller  in Fig. 5 than that in Fig. 6, while the reverse is true for the gap
between $l_{1}=0.5$ and $l_{1}=0.1$.
This implies that the performance improvement to increase the priority for
primary channel estimation is more significant when  
$\tau_{\mathrm{total}}$ is  relatively small.
Thus, when the pilot length  $\tau_{\mathrm{total}}$ is severely limited, 
higher priority should 
be given to the  training phase 1.
This is expected since the estimation of
the direct-link channels in the first phase affects not only
the primary communication rate, but also the quality of the
channel estimation of the backscatter channels.


\begin{figure}[!t]
	\centering
\includegraphics[height=3in, width=3.809in]{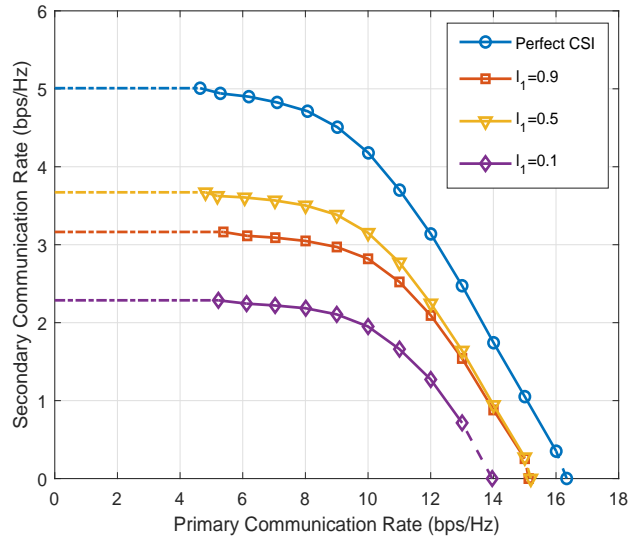}
\caption{The achievable rate-region of cell-free symbiotic radio system, 
where   $\tau_{\mathrm{total}}=50$. }
\end{figure}
\begin{figure}[!t]
	\centering
\includegraphics[height=3in, width=3.809in]{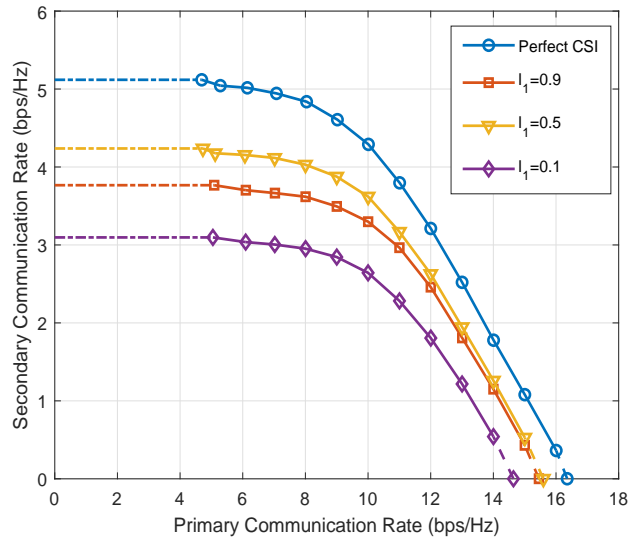}
\caption{The achievable rate-region of cell-free symbiotic radio system, 
where   $\tau_{\mathrm{total}}=100$. }
\end{figure}

\section{Conclusion}

In this paper, a novel cell-free symbiotic radio system was investigated, 
in which a number of distributed APs cooperatively send primary information to the receiver, 
while concurrently supporting the secondary backscattering communication. 
The achievable rates of both the active primary and passive secondary communications  were first  derived 
under the assumption of perfect CSI. Furthermore, 
a two-phase uplink-training based channel
estimation method was proposed to  effectively estimate the direct-link channel and cascaded 
backscatter channel, and the achievable rates were revisited
with  channel estimation errors.
In order to characterize the achievable rate-region, 
a beamforming optimization problem was formulated to maximize
the passive secondary communication rate with a targeting active primary 
communication rate constraint, for both perfect CSI
and imperfect CSI. Efficient algorithms were proposed 
to solve the formulated optimization problem.
The performance of the cell-free symbiotic
radio communication systems was 
validated with extensive simulation results.  

 

%


 
\begin{appendices}
 
\vspace{10pt} 

\section{Proof of Theorem 1}
To prove Theorem 1, we first consider a relaxed problem of (\ref{eq:OP2}) 
by omitting  the targeting primary communication
rate constraint (\ref{eq:OP2}b). The problem is formulated as

\begin{subequations}\label{eq:OPMRT1}
  \begin{align}
  \begin{split}
    \max   \limits_{ {\mathbf{w}_m, m=1,...,M}} \quad &
    | \sum_{m=1}^{M}
    \mathbf{f}_{m}^{H} {\mathbf{w}_{m}}|^{2}
  \end{split}\\
  \begin{split}
    {\rm s.t.} 
    \quad  & \big \| {\mathbf{w}_m} \big \|^{2} \leq  P_{m},  \qquad  m=1,...,M.
  \end{split}
\end{align}
\end{subequations}

It is not difficult to see that the optimal solution to problem (\ref{eq:OPMRT1})  is 
the per-AP maximum ratio transmission (MRT) beamforming  with maximum transmit power, 
which is given by (\ref{eq:wMRT}) in Theorem 1. In this case, the corresponding  
primary communication rate on the left hand side of (\ref{eq:OP2}b) 
is given in $\widehat{R}_{s}$ shown in (\ref{eq:Rth<widehideR}).



Therefore, when the condition $R_{\mathrm{th}}\leq\widehat{R}_{s}$ given in Theorem 1
is satisfied, 
the solution in (\ref{eq:wMRT}) satisfies both constraints in (\ref{eq:OP2}b) and (\ref{eq:OP2}c), 
and thus it is also feasible to the original problem (\ref{eq:OP2}). 
Furthermore, since the beamforming vectors in (\ref{eq:wMRT}) is the optimal solution 
to the relaxed problem (\ref{eq:OPMRT1}) and is also feasible to (\ref{eq:OP2}),
it must also be the optimal solution to
the original problem (\ref{eq:OP2}), since the former has a larger feasibility 
region than the latter. 

This completes the proof of Theorem 1.

\vspace{10pt}

\section{Proof of Lemma 1}
To prove Lemma 1, we note that at the $(l+1)$-th iteration, 
with Step 4 of Algorithm 2, we have
$\mathbf{w}^{(l+1)}=\mathbf{w}^{\star(l)}$, where $\mathbf{w}^{\star(l)}$ 
represents the optimal solution at the $(l)$-th iteration.
Therefore, the lower bound (\ref{eq:SCA}) holds with 
equality at the local point $\mathbf{w}=\mathbf{w}^{\star(l)}$, i.e., 
\begin{equation}\label{eq:equal}
  \begin{aligned}
    F(\mathbf{w}^{\star(l)})&=| \mathbf{f  }^{H}  {\mathbf{w}^{\star(l)}}|^{2}
    =F_{ {\rm low} }({\mathbf{w}}^{\star(l)}|{\mathbf{w}^{\star(l)}})\\ &
    =F_{ {\rm low} }({\mathbf{w}}^{\star(l)}|{\mathbf{w}^{(l+1)}}),\forall l.  
  \end{aligned}
\end{equation}

Furthermore, with (\ref{eq:SCA}) and Algorithm 2, 
we obtain the following relationship  

\begin{equation}\label{eq:relationship}
  \begin{aligned}
    F(\mathbf{w}^{\star(l+1)}) &  \underset{(a)}{\geq} F_{ {\rm low} }({\mathbf{w}}^{\star(l+1)}|{\mathbf{w}^{(l+1)}})
        \\ &\underset{(b)}{\geq} F_{ {\rm low} }({\mathbf{w}}^{\star(l)}|{\mathbf{w}^{(l+1)}}) 
        \underset{(c)}{=}F(\mathbf{w}^{\star(l)}),
  \end{aligned}
\end{equation}
where the inequality $(a)$ follows from the global lower bound (\ref{eq:SCA}), 
the inequality  $(b)$ 
holds since $\mathbf{w}^{\star(l+1)}$ is the optimal solution of problem (\ref{eq:OPSCA3}) 
in the $(l+1)$-th iteration, and the equality$(c)$ follows from (\ref{eq:equal}).

Thus, the relationship (\ref{eq:relationship}) shows that the objective  value $F(\mathbf{w}^{\star(l)})$ of (\ref{eq:OPSCA2}) obtained 
in Algorithm 2 is monotonically non-decreasing, and hence
converges to a finite limit. Besides, note that the lower
bound $ F_{ {\rm low} }({\mathbf{w}}|{\mathbf{w}^{(l+1)}})$ has 
identical gradient as $F(\mathbf{w})$ at the local point $\mathbf{w}^{\star(l)}$,
i.e., 
\begin{equation} 
  \nabla F(\mathbf{w}^{\star(l)})=
   \nabla F_{ {\rm low} }({\mathbf{w}}^{\star(l)}|{\mathbf{w}^{(l+1)}}). 
\end{equation}

Then the sequence ${\mathbf{w}}^{\star(l)}$ converges to a point that fulfills
the KKT optimality conditions of the original non-convex 
problem (\ref{eq:OPSCA2}) \cite{zeng2017energy,Zappone2017a,Marks1978a}.

This completes the proof of Lemma 1.

\end{appendices}

\ifCLASSOPTIONcaptionsoff
  \newpage
\fi

\end{document}